\documentclass[12pt]{article}
\usepackage{amscd,amsmath,amssymb,amstext,amsthm,exscale,latexsym}
\usepackage{graphicx}
\newcommand {\g }   {\gamma}       
\newcommand {\dl}   {\delta}       
        
\newcommand {\ve}   {\varepsilon}

\newcommand {\vf }  {\varphi}      
         \newcommand {\om}  {\omega}

\newcommand {\pl}   {\partial}     \newcommand {\nb}  {\nabla}
\renewcommand {\sin}{{\sf\,sin\,}}       \renewcommand {\cos}{{\sf\,cos\,}}

       \renewcommand {\exp}{{\sf\,exp\,}}

       \renewcommand {\lim}{{\sf\,lim\,}}
\newcommand   {\ex}{{\sf\,e}}

     \newcommand   {\diag}{{\sf\,diag\,}}
         \newcommand   {\tr}{{\sf\,tr\,}}

             \renewcommand   {\P}{{\sf P}}

\newcommand {\MO}  {{\mathbb O}}   
   \newcommand {\MR}  {{\mathbb R}}
\newcommand {\MS}  {{\mathbb S}}   
\newcommand {\MU}  {{\mathbb U}}   
   
   \newcommand {\MZ}  {{\mathbb Z}}

\newcommand {\Go}  {\mathfrak{o}}   
   
\newcommand {\Gs}  {\mathfrak{s}}

\newcommand {\Sc}  {{\textsc{c}}}

\newcommand {\Ss}  {{\textsc{s}}}

\newcommand {\one}  {1\!\!1}

\newtheorem{theorem}{Theorem}[section]

\theoremstyle{definition}

\begin{document}
\title     {Point disclinations in the Chern--Simons geometric theory of
            defects}
\author    {M.~O.~Katanaev
            \thanks{E-mail: katanaev@mi.ras.ru} and
            B.~O.~Volkov
            \thanks{E-mail: borisvolkov1986@gmail.com}\\ \\
              \sl Steklov Mathematical Institute,\\
            \sl ul. Gubkina, 8, Moscow, 119991, Russia}

\maketitle
\begin{abstract}

We use the Chern--Simons action for a $\MS\MO(3)$-connection for the description
of point disclinations in the geometric theory of defects. The most general
spherically symmetric $\MS\MO(3)$-connection with zero curvature is found. The
corresponding orthogonal spherically symmetric $\MS\MO(3)$ matrix and $n$-field
are computed. Two examples of point disclinations are described.
\end{abstract}
\section{Introduction}
The geometric theory of defects [1--4]
\nocite{KatVol92,Katana05,Katana13B,Katana17C}
describes dislocations (defects in elastic media) and disclinations (defects in
the spin structure of media) in the framework of the Riemann--Cartan geometry.
The curvature and torsion two-forms are surface densities of the Burgers and
Frank vectors, respectively. There are many examples of dislocations described
in the framework of the geometric theory of defects \cite{Katana05}, but only a
few disclinations. As far as we know, the first example of a  straight linear
disclination described within the geometric theory of defects is given in
\cite{Katana17C}. In the present paper, we give another examples of
disclinations. Now these are point disclinations.

The geometric theory of defects is well suited for the description of single
defects as
well as their continuous distribution. The only variables in the theory are
Cartan variables: a vielbein field $e_\mu{}^i$ and $\MS\MO(3)$-connection
$\om_\mu{}^{ij}=-\om_\mu{}^{ji}$, where $\mu,\nu,\dotsc=1,2,3$ and
$i,j,\dotsc=1,2,3$ are space and internal indices, respectively. We consider the
case of the Euclidean vielbein $e_\mu{}^i:=\dl_\mu^i$, which means the absence
of elastic stresses in the media. Then we have the ordinary $\MS\MO(3)$ gauge
Yang--Mills theory, $A_\mu{}^{ij}:=\om_\mu{}^{ij}$, living in the flat
three-dimensional Euclidean space (we consider only the static case). For
single disclinations we have the curvature singularity at the core of disclination
and zero curvature outside the core. Therefore the Chern--Simons action is well
suited for the description of single disclinations because it produces the zero
curvature equations of equilibrium for the $\MS\MO(3)$-connection. The respective
equations were solved for a single straight linear disclination in
\cite{Katana17C}. In the present paper, we find the most general
spherically symmetric $\MS\MO(3)$-connection with zero curvature. It depends
on one arbitrary function on radius. There are no disclinations for the
particular choice of this function. In a general situation, point disclinations
are present. We consider two examples. The first is the hedgehog spherically
symmetric disclination. The second is a point disclination with essential
singularity at the origin and constant $n$-field at infinity.
\section{Chern--Simons action and point disclinations}
We consider a three-dimensional Euclidean space with Cartesian coordinates
$(x^\mu)\in\MR^3$, $\mu=1,2,3$. Let components
$A_\mu{}^{ij}(x)=-A_\mu{}^{ji}(x)$, $i,j=1,2,3$ of the $\MS\MO(3)$-connection
local form be defined on this space (in other words, the $\MS\MO(3)$ Yang-Mills
fields). From a geometric point of view, we have a topologically trivial
manifold $\MR^3 $ with a given Riemann-Cartan geometry defined by the triad
field $e_\mu{}^i$, satisfying equality
$\dl_{\mu\nu}=e_\mu{}^ie_\nu{}^j\dl_{ij}$, where $\dl_{ij}:=\diag(+++)$ is the
 Euclidean metric, and the $\MS\MO(3)$-connection
$\om_{\mu i}{}^j=A_{\mu i}{}^j$.

The curvature and torsion have usual expressions in terms of Cartan variables
\begin{equation}                                                  \label{qcnftl}
\begin{split}
  R_{\mu\nu}{}^{ij}=&\pl_\mu\om_\nu{}^{ij}-\pl_\nu\om_\mu{}^{ij}
  -\om_\mu{}^{ik}\om_{\nu k}{}^j+\om_\nu{}^{ik}\om_{\mu k}{}^j,
\\
  T_{\mu\nu}{}^i=&\pl_\mu e_\nu{}^i-\pl_\nu e_\mu{}^j-e_\mu{}^j\om_{\nu j}{}^i
  +e_\nu{}^j\om_{\mu j}{}^i.
\end{split}
\end{equation}
Flat vielbein means that elastic stresses are absent in media.

For the description of point disclinations, we choose the Chern--Simons action
\cite{CheSim74}
\begin{equation}                                                  \label{ubbcng}
  S_{\Sc\Ss}:=\int_{\MR^3}\!\!\!\tr\left(dA\wedge A-\frac23A\wedge
  A\wedge A\right),
\end{equation}
where $A:=(A_i{}^j)=(dx^\mu A_{\mu i}{}^j)$ is a local $\MS\MO(3)$-connection
1-form. Variation of action (\ref{ubbcng}) with respect to connection yields
the zero curvature equilibrium equations:
\begin{equation}                                                  \label{unnbch}
  F_{\mu\nu i}{}^j:=\pl_\mu A_{\nu i}{}^j-\pl_\nu A_{\mu i}{}^j
  -A_{\mu i}{}^kA_{\nu k}{}^j+A_{\nu i}{}^kA_{\mu k}{}^j=0.
\end{equation}
That is the connection must be flat. The Chern--Simons action (\ref{ubbcng}) is
well suited for the description of point disclinations because the curvature
must be zero outside the core of disclinations. In what follows, we assume that
a disclination is located at the origin of coordinate system.

Components of the connection can be parameterized by the field with two
indices:
\begin{equation}                                                  \label{unnxbg}
  A_\mu{}^{ij}=A_\mu{}^k\ve_k{}^{ij},\qquad
  A_\mu{}^k:=\frac12A_\mu{}^{ij}\ve^k{}_{ij}.
\end{equation}
where $\ve_{ijk}$, $\ve_{123} = 1$, is the totally antisymmetric tensor.
Then components of the curvature local form of the $\MS\MO(3)$-connection are
\begin{equation}                                                  \label{ubnhyf}
  F_{\mu\nu k}:=\frac12F_{\mu\nu}{}^{ij}\ve_{ijk}
  =\pl_\mu A_{\nu k}-\pl_\nu A_{\mu k}+A_\mu{}^i A_\nu{}^j\ve_{ijk}.
\end{equation}

Now we find the most general spherically symmetric solution to Eq.\
(\ref{unnbch}).

We assume that the global rotation group $\MS\MO(3)$ acts simultaneously both on
the base $\MR^3$, and on the Lie algebra $\Gs\Go(3)$,  which, as a vector
space, is also a three-dimensional Euclidean space $\MR^3$. It means that if
$S\in\MS\MO(3)$ is an orthogonal matrix, then the transformation
has the form
\begin{equation*}
  A_\mu{}^{ij}\mapsto S^{-1\nu}_{~~\mu} A_\nu{}^{kl}S_k{}^iS_l{}^j,\qquad
  S\in\MS\MO(3).
\end{equation*}

Under this assumption, the difference between Greek and Latin indices
disappears, but we shall, as far as possible, distinguish them.

If  we  include  reflections into the rotation group, then $A_\mu{}^{ij}$ become
components of the third rank tensor with respect to the action of the full
rotation group $\MO(3)$, and $A_\mu{}^k$ become components of the second rank
pseudo-tensor, due to the presence of the third rank pseudo-tensor $\ve_{ijk}$.

Now the most general spherically symmetric components of the connection have the
form
\begin{equation}                                                  \label{unbcht}
  A_\mu{}^i(x)=\ve_\mu{}^{ij}\frac{x_j}rW(r)+\dl_\mu^iV(r)
  +\frac{x_\mu x^i}{r^2}U(r),\qquad r\ge0,
\end{equation}
where $W$, $V$, $U$ are arbitrary sufficiently smooth functions of radius
$r:=\sqrt{x^\mu x_\mu}$. These functions are defined only on the non-negative
semi-axis $r\ge0$. Under the action of the full rotation group $\MO(3)$
the function $W$ is a scalar, and $V$ and $U$ are pseudoscalars.

The Lie algebras of $\MS\MO(3)$ and $\MS\MU(2)$ groups are the same, and
$\MS\MU(2)$ Yang--Mills model minimally interacting with the triplet of scalar
fields in the adjoint representation coincides formally with the $\MS\MO(3)$
Yang--Mills model minimally interacting with the triplet of scalar fields in the
fundamental representation. There is the famous 't Hooft--Polyakov monopole
solution \cite{tHooft74,Polyak74} to these models. It is spherically symmetric
and corresponds to ansatz (\ref{unbcht}) with $V=U=0$.

Direct computations of components of the spherically symmetric curvature tensor
lead to the following expression
\begin{multline}                                                  \label{uxnbhg}
  F_{\mu\nu}{}^i=\ve_{\mu\nu}{}^i\left[W'+\frac Wr+V(V+U)\right]
  +\frac{\ve_{\mu\nu}{}^jx_jx^i}{r^3}\big(W-rW'+rW^2-rVU\big)+
\\
  +\frac{x_\mu\dl_\nu^i-x_\nu\dl_\mu^i}{r^2}\big[rV'-U-rW(V+U)\big],
\end{multline}
where the prime mark denotes differentiation by the radius. Now, just like for
t'Hooft-Polyakov monopole, we introduce the dimensionless function $K(r)$ as
\begin{equation*}
  W:=\frac{K-1}r.
\end{equation*}

Then the expression for curvature~(\ref{uxnbhg}) becomes simpler:
\begin{multline}                                                  \label{ubcnhy}
  F_{\mu\nu}{}^i=\frac{\ve_{\mu\nu}{}^i}r\big[K'+rV(V+U)\big]
  +\frac{\ve_{\mu\nu}{}^jx_jx^i}{r^3}\left(-K'+\frac{K^2-1}r-rVU\right)+
\\
  +\frac{x_\mu\dl_\nu^i-x_\nu\dl_\mu^i}{r^2}\big[rV'-U-(K-1)(V+U)\big].
\end{multline}
Equilibrium equations~(\ref{unnbch}) in the spherically symmetric case
yield the following system of equations
\begin{align}                                                     \label{unbhhu}
  K'+rV(V+U)=&0,
\\                                                                \label{unmvjo}
  -K'+\frac{K^2-1}r-rVU=&0,
\\                                                                \label{ukfoij}
  rV'-U-(K-1)(V+U)=&0,
\end{align}
because tensor structures in Eq.\ (\ref{ubcnhy}) are functionally independent.
Thus, in the spherically symmetric case, the equilibrium equations are reduced
to the system of three nonlinear ordinary differential equations on three
arbitrary functions $K$, $V$ and $U$. We shall see below, that these equations
are dependent and therefore a general solution of this system contains a
functional arbitrariness.
\begin{theorem}
A general solution of the system of equations (\ref{unbhhu})--(\ref{ukfoij})
has the form
\begin{align}                                                     \label{ubbxvl}
  K=&\cos f,
\\                                                                \label{ujdfki}
  V=&\pm\frac{\sin f}r,
\\                                                                \label{ujgytr}
  U=&\pm\frac{rf'-\sin f}r,
\end{align}
where $f(r)$ is an arbitrary sufficiently smooth function of radius $r\ge0$ and
the upper or lower signs should be selected simultaneously.
\end{theorem}
\begin{proof}
Add equations~(\ref{unbhhu}) and~(\ref{unmvjo}):
\begin{equation*}
  rV^2+\frac{K^2-1}r=0.
\end{equation*}
It implies
\begin{equation}                                                  \label{unncmj}
  V=\pm\frac{\sqrt{1-K^2}}r.
\end{equation}
The inequality $|K|\le0$ is necessary and sufficient for the existence of real
solutions of this equation. Therefore, without loss of generality, we put
$K:=\cos f$, where $f(r)$, is a sufficiently smooth function. After
that, Eq.~(\ref{unncmj}) goes to Eq.~(\ref{ujdfki}). Now Eq.~(\ref{unmvjo})
implies expression for $U$ (\ref{ujgytr}). Afterwards, one can verify that the
third Eq.~(\ref{ukfoij}) is automatically satisfied for an arbitrary function
$f$.
\end{proof}

Thus, a general spherically symmetric solution of the Euler--Lagrange equations
(\ref{unnbch}) has the form
\begin{equation}                                                  \label{uncbhf}
  A_\mu{}^i(x)=\frac{\ve_\mu{}^{ij}x_j}{r^2}(\cos f-1)+\dl_\mu^i\frac{\sin f}r
  +\frac{x_\mu x^i}{r^3}(rf'-\sin f),
\end{equation}
where $f(r)$ is an arbitrary sufficiently smooth function. This connection is
flat and satisfies the zero curvature equation everywhere in $\MR^3$ except,
possibly, the origin of coordinate system.

If $f$ is a smooth function and sufficiently fast goes to zero as $r\to0$,
then the curvature of the $\MS\MO(3)$-connection is identically zero on the
whole $\MR^3$, and there are no disclinations. If at the origin $f(0)\ne0$, then
the connection and the curvature can be singular at the origin, and
disclinations may appear. To find their structure, we have to reconstruct the
$n$-field.
\subsection{Disclinations}
The spin structure of media, e.g.\ ferromagnets, is described by the unit vector
field $n(x)=\big(n^i(x)\big)$, $n^2:=n^in_i=1$. In the geometric theory of
defects \cite{KatVol92,Katana05,Katana13B}, the unit vector field is
parameterized by orthogonal matrix:
\begin{equation}                                                  \label{ubndhy}
  n^i(x)=n_0^jS_j{}^i(x),\qquad S\in\MS\MO(3),
\end{equation}
where $n_0$ is some fixed unit vector. If the fields $n$ and $S$ are continuous,
then defects are absent. By definition, disclinations are discontinuities in the
unit vector field $n$. In that case matrix $S$ is also discontinuous. The
inverse statement may be not true. There may exist situations when matrix $S$ is
discontinuous but $n$-field is continuous, for instance, the hedgehog
disclination considered in the next section. For finite number of disclinations
the unit vector field exists everywhere except the cores
of disclinations, where it has discontinuities. In the limiting case for
continuous distribution of disclinations, the $n$-field does not exist at all
and cannot be used for describing defects. The advantage of the geometric theory
of defects is that the basic variable is the $\MS\MO(3)$-connection which exists
even in the absence of the $n$-field. The criteria for the presence of
disclinations is nonzero curvature of $\MS\MO(3)$-connection. If curvature is
zero, then the $\MS\MO(3)$-connection is a pure gauge, which implies the
existence of the orthogonal matrix $S$ and, consequently, the $n$-field.

In our case, the curvature of spherically symmetric connection (\ref{uncbhf}) is
zero everywhere except, possibly, the origin of coordinate system. So, the
connection (\ref{uncbhf}) may describe point disclinations located at the
origin. To reconstruct the $n$-field corresponding to connection (\ref{uncbhf}),
we must find the orthogonal matrix $S(x)$ where it exists.

In those  connected and simply connected open subsets of the Euclidean space,
where the curvature is zero, the connection is a pure gauge
\begin{equation}                                                  \label{ubvcgf}
  A_{\mu i}{}^j=A_\mu{}^k\ve_{ki}{}^j=\pl_\mu S^{-1k}_{~~i}S_k{}^j,\qquad
  S=(S_{i}{}^j)\in\MS\MO(3).
\end{equation}
Our aim is to find matrix $S$ for a given connection~(\ref{uncbhf}).
The equation for $S$ is
\begin{equation*}
  \pl_\mu S^{-1}=A_\mu S^{-1},
\end{equation*}
where matrix indices are omitted. This equation has straightforward meaning in
the geometric theory of defects. Namely, consider the covariant derivative of
the $n$-field (\ref{ubndhy}):
\begin{equation*}
  \nb_\mu n^i:=\pl_\mu n^i+n^jA_{\mu j}{}^i
  =n_0^k(\pl_\mu S_k{}^i+S_k{}^jA_{\mu j}{}^i).
\end{equation*}
For pure gauge (\ref{ubvcgf}) it is zero. This means that in domains with
zero curvature the $n$-field is obtained by parallel transport of vector $n_0$
with the pure gauge connection. The result of parallel transportation does not
depend on curves of transportation, because curvature is zero.

If the curvature is zero, the parallel transport is independent of the curve,
along which it is transported. Therefore, we consider an arbitrary curve
$\g=y(t)\in\MR^3\setminus 0$, $t\in[0,b]$ with the starting point $y_0:=y(0)$
and ending point $y_b:=y(b)$. Then for matrix $S$ we get the ordinary
differential equation along $\g$:
\begin{equation}                                                  \label{ujfkiy}
  \dot S^{-1}=\dot y^\mu A_\mu S^{-1}.
\end{equation}
When a curve passes through a point $y(t)$, the solution of this equation
is given by the path-ordered exponent:
\begin{equation}                                                  \label{unbhcy}
  S^{-1}\big(y(t)\big)=\P\exp\left(\int_0^t\!\!\!ds\,\dot
  y^\mu(s)A_\mu(y(s))\right)S^{-1}_0,
\end{equation}
where $S_0$ is a fixed matrix at the starting point $y_0$.

Let  $\g=\big(y^\mu(t):=x^\mu t\big)$, $t\in[1,\infty]$, be a straight half-line
(ray) with the starting point $x\in\MR^3\setminus0$ and the end at infinity.
We assume that  $S_0:=S(\infty):=\one$. It means that we consider connections
$A(x)$ tending to zero as $x\to \infty$. Then for connection~(\ref{ubvcgf}),
the equality
\begin{equation}                                                  \label{ubxnyt}
  \dot y^\mu A_{\mu i}{}^j=f'x^k\ve_{ki}{}^j
\end{equation}
is satisfied. Now it is easily verified, that the matrices under the integral
in the ordered exponent commute:
\begin{equation*}
  [\dot y^\mu A_\mu,\dot y^\nu A_\nu]=0.
\end{equation*}
Hence, the path-ordered exponent coincides with the ordinary one, and
integral~(\ref{unbhcy}) can be easily computed:
\begin{equation*}
  \int_\infty^1\!\!\!ds\,\dot y^\mu A_{\mu i}{}^j
  =\int_\infty^1\!\!\!ds\,f'x^k\ve_{ki}{}^j
  =\int_\infty^1\!\!\!ds\,\frac{df}{d(rs)}x^k\ve_{ki}{}^j
  =\frac{x^k\ve_{ki}{}^j}r\big[f(r)-f(\infty)\big].
\end{equation*}

We emphasize that this integral is independent of the choice of the curve $\g$
with starting point at infinity and ending point at $x$, because the curvature
of $\MS\MO(3)$-connection is zero.

That is, the solution of the Eq.~(\ref{ujfkiy}) is
\begin{equation}                                                  \label{unnvmf}
  S^{-1j}_{~~i}=\exp(-f^k\ve_{ki}{}^j),\qquad
  f^k:=\frac{x^k}r\big[f(\infty)-f(r)\big].
\end{equation}
The vector $(f^k)$ is an element of the Lie algebra $\Gs\Go(3)$. Its direction
coincides with the axis of the rotation in the isotopic space and its length is
equal to the rotation  angle.
The exponential map for the $\MS\MO(3)$ group is well-known:
\begin{equation}                                                  \label{unncbg}
  S_i{}^j=\exp(f^k\ve_{ki}{}^j)
  =\dl_i^j\cos F+\frac{f^k\ve_{ki}{}^j}F\sin F
  +\frac{f_if^j}{F^2}(1-\cos F),
\end{equation}
where $F^2:=f^if_i=\big[f(\infty)-f(r)\big]^2$. Note, that Eq.~(\ref{unncbg})
is valid both for positive and negative $F$.

Thus, we have found a general form of the orthogonal matrix for the spherically
symmetric connection (\ref{uncbhf}). Afterwards the $n$-field is defined by
Eq.~(\ref{ubndhy}).
\subsection{Examples of point disclinations}
The orthogonal matrix~(\ref{unncbg}) is determined by the difference
$f(\infty)-f(r)$, where $f(r)$ is an arbitrary sufficiently smooth function.
Without loss of generality, we put $f(\infty)=0$ and change the sign of $f(r)$.
Then we choose $F(r)=f(r)$. Afterwards the spherically symmetric orthogonal
matrix takes the form
\begin{equation}                                                  \label{unnpbg}
  S_i{}^j(x)=\exp(f^k\ve_{ki}{}^j)
  =\dl_i^j\cos f+\frac{f^k\ve_{ki}{}^j}f\sin f
  +\frac{f_if^j}{f^2}(1-\cos f),
\end{equation}
where
\begin{equation}                                                  \label{unbcgf}
  f^i=\frac{x^i}r f(r),
\end{equation}
$f(r)$ being an arbitrary function.
This orthogonal matrix is clearly spherically symmetric as it should be.

{\bf Hedgehog disclination.}
\begin{figure}[hbt]
\hfill\includegraphics[width=.4\textwidth]{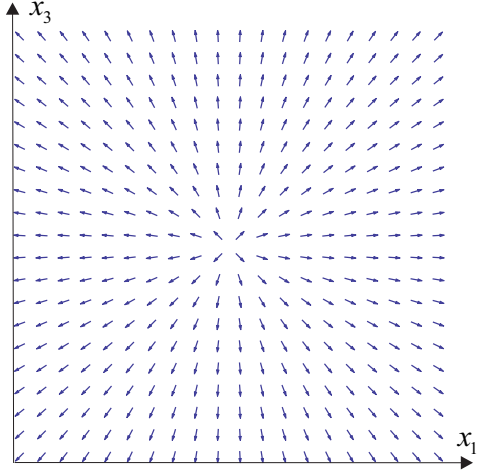}
\hfill {}
\centering\caption{Hedgehog disclination. The section $x_2=0$ is shown in the
figure.}
\label{fdisclhedgehog}
\end{figure}
The hedgehog disclination is the spherically symmetric distribution of the
$n$-field
$$
n^i(x)=\frac{x^i}r
$$
with the singularity at the origin. Its section $x_2=0$ is shown in
Fig.~\ref{fdisclhedgehog}.

The analysis in the last section tells us that we must have fixed vector $n_0$
at infinity to make the parallel transport, but it clearly breaks the spherical
symmetry. Therefore, this
disclination cannot be described by spherically symmetric orthogonal matrix
(\ref{unnpbg}) starting with the fixed vector at infinity. Hence, we do the
following. First we rotate the vector $n_0$ at infinity by the orthogonal matrix
to make it spherically symmetric and afterwards apply rotation (\ref{unnpbg}).

Let us consider the spherical coordinate system and the orthogonal spherically
symmetric basis $(e_{\hat r},e_{\hat\theta},e_{\hat\vf})$, where
\begin{equation*}
  e_{\hat r}:=\pl_r,\qquad e_{\hat\theta}:=\frac1r\pl_\theta,\qquad
  e_{\hat\vf}:=\frac1{r\sin\theta}\pl_\vf.
\end{equation*}
Bellow Latin indices $i,j,\dotsc$ run through  $(\hat r,\hat\theta,\hat\vf)$.

Let vector $n_0$ at infinity be directed along the $z$-axis: $n_0:=(0,0,1)$ in
Cartesian coordinates. In the spherical coordinates we have
\begin{equation*}
  n_0=n_0^{\hat r}e_{\hat r}+n_0^{\hat\theta}e_{\hat\theta}
  +n_0^{\hat\vf}e_{\hat\vf}=\cos\theta e_{\hat r}-\sin\theta e_{\hat\theta},
\end{equation*}
Now we make rotation by orthogonal matrix in the spherical coordinates
\begin{equation*}
  (P_i{}^j):=\begin{pmatrix} ~~\cos\theta & \sin\theta & 0 \\
   -\sin\theta & \cos\theta & 0 \\ 0 & 0 & 1 \end{pmatrix},
\end{equation*}
which is not spherically symmetric.
Then the vector at infinity becomes
\begin{equation}
\label{inftyrot}
  \tilde n(\theta,\vf):=(n_0^j P_j{}^i)=e_{\hat r}.
\end{equation}
This distribution is clearly spherically symmetric and coincides with that of
the hedgehog disclination at infinity.

Now we make the vector field on $\MR^3$ by applying spherically symmetric matrix
(\ref{unnpbg}):
\begin{equation*}
  n^i(x)=\tilde n^jS_j{}^i(x)=n_0^kP_k{}^jS_j{}^i=\frac{x^i}r.
\end{equation*}
We see that $n$-field is everywhere directed along radius and has unit length.
So, it describes the hedgehog disclination. It is obtained by applying the
matrix $PS$ to the vector $n_0$ at infinity. This matrix is not spherically
symmetric but orthogonal. It depends on the choice of the smooth
function $f$ such that $f(\infty)=0$. Particularly, if we choose $f\equiv 0$,
then $S\equiv\one$ on the whole $\mathbb{R}^3$.

Alternatively, we can say, that the vector field $n(x)$ of the hedgehog
disclination is obtained by the parallel transport of the vector $n_0$ at the
north pole at infinity by the flat connection
\begin{equation*}
  A_\mu:=\pl_\mu(PS)^{-1}(PS).
\end{equation*}
This connection is not spherically symmetric because of the matrix $P$.

By going back to Cartesian coordinates, one can see that the matrix $PS$ is
defined on the whole $\MR^3$ except the nonpositive part of the $z$-axis. This
singularity cannot be removed, because otherwise we are in contradiction with
the well known Hairy Ball Theorem \cite{EisGuy79}.
\qed

Now consider spherically asymmetric disclinations. Fix the vector at infinity as
$n_0:=(0,0,1)$ in Cartesian coordinates $(x^1,x^2,x^3)$. Thus the spherical
symmetry is broken. Then components of the $n$-field
\begin{equation*}
  n^i(x):=n_0^j S_j{}^i(x),
\end{equation*}
where matrix $S$ is given by Eq.~(\ref{unnpbg}), are
\begin{equation}                                                  \label{ubbcvf}
\begin{split}
  n_1=&-\frac{x_2}r\sin f+\frac{x_1x_3}{r^2}(1-\cos f),
\\
  n_2=&~~\frac{x_1}r\sin f+\frac{x_2x_3}{r^2}(1-\cos f),
\\
  n_3=&\cos f+\frac{x_3^2}{r^2}(1-\cos f).
\end{split}
\end{equation}
Here coordinate indices, for convenience, are lowered to distinguish them from
exponents. These components of $n$-field in spherical coordinates take the
form
\begin{equation}                                                  \label{ummvju}
\begin{split}
  n_1=&-\sin\theta\sin\vf\sin f+\sin\theta\cos\theta\cos\vf(1-\cos f),
\\
  n_2=&~~\sin\theta\cos\vf\sin f+\sin\theta\cos\theta\sin\vf(1-\cos f),
\\
  n_3=&~~\cos f+\cos^2\theta(1-\cos f).
\end{split}
\end{equation}
It means, that the limit of the $n$-field at $r\to0$, which does not depend on
the curve approaching the origin, exists if and only if $f(0)=k\pi$, $k\in\MZ$.
This is the exceptional case, when $n$-field is continuous at zero, and
disclinations do not appear. If $f(0)\ne k\pi$, then the limit of the $n$-field
at $r\to 0$ depends on the path to the origin. Consequently, in a general case,
the origin is an essential singularity, and the model describes point
disclinations at the origin.

After fixing the vector $n_0:=(0,0,1)$, there remains the invariance with
respect to rotations in the $x_1,x_2$ plane. Therefore, for a visual presentation
of disclinations it suffices to put $x_2=0$, i.e., to study distributions of
$n$-field in the plane $x_1,x_3$:
\begin{equation}                                                  \label{unmfju}
\begin{split}
  n_1=&\frac{x_1x_3}{r^2}(1-\cos f),
\\
  n_2=&\frac{x_1}r\sin f,
\\
  n_3=&\cos f+\frac{x_3^2}{r^2}(1-\cos f).
\end{split}
\end{equation}
We see that in a general case vector $n$ has a nonzero component in the
direction perpen\-dicular to the $x_1,x_3$ plane, which makes visualization more
complicated.

Different distributions of $n$-field depend on the choice of the function
$f(r)$. Set $f(\infty)=0$, i.e., the $n$-field coincides with $n_0$ at infinity.
If $f(0)=k\pi$, the unit vector field is continuous at zero and disclinations
are absent. Otherwise there are disclinations with an essential singularity at
the origin.
\begin{figure}[hbt]
\hfill\includegraphics[width=.9\textwidth]{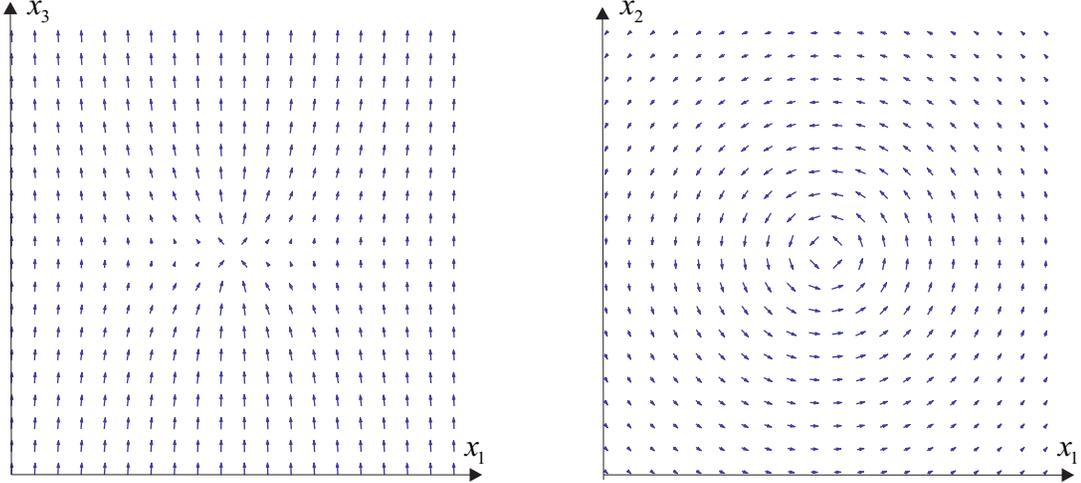}
\hfill {}
\centering\caption{Two sections $x_2=0$ and $x_3=0$ of the disclination with
$f(r):=\pi\ex^{-r}/2$. Arrows are projections of the vector $n$ on the
corresponding plane. If the length of an arrow is smaller then unity, then it
means that the vector has the component in the perpendicular direction. The
spherical symmetry is broken by the boundary condition $n_0:=(0,0,1)$.}
\label{fdisclxzxy}
\end{figure}

{\bf Example two.}
Let
\begin{equation*}
  f(r):=\frac\pi2\ex^{-r},\qquad\Rightarrow\qquad
  f(0)=\frac\pi2,\qquad f(\infty)=0.
\end{equation*}

In this case, the vector field $n$ (\ref{unmfju}) in the plane $x_2=0$ has all
three non-trivial components. Figure \ref{fdisclxzxy} shows projections of
$n$-field on two planes: $x_2=0$ and $x_3=0$. At infinity, the projection of a
vector field onto the plane $x_2=0$ has the unit length, because the
perpendicular component is absent. In interior points, the projection is
smaller due to the appearance of the perpendicular component. The projection of
vectors $n$ onto the plane $x_3=0$, on the contrary, is zero at infinity and
nontrivial at interior points. It is clearly seen in Fig.~\ref{fdisclxzxy},
which is drawn numerically.
\qed
\section{Conclusion}
We have used the Chern--Simons action for $\MS\MO(3)$-connection for describing
point discli\-nations in elastic media with unit $n$-field within the geometric
theory of defects. The metric and corresponding vielbein are assumed to be
Euclidean. The nontrivial geometry arises due to nontrivial
$\MS\MO(3)$-connection which has singularity at one point and is flat outside.
The most general spherically symmetric flat connection is found. It depends on
one arbitrary function of radius. Depending on this function, various
possibilities arise. Two examples are considered. The one with spherical
symmetry describes the hedgehog disclination. For constant boundary condition
for $n$-field at infinity which breaks the spherical symmetry, the connection
and curvature, in general, has essential singularity at the origin and describe
point disclinations. One example of such disclination is explicitly constructed.

This work is supported by the Russian Science Foundation under grant
19-11-00320.

\end{document}